\documentclass[a4paper]{article}
\usepackage[a4paper, total={6.5in, 9in}]{geometry}
\usepackage{bm}
\usepackage{float}
\usepackage[T1]{fontenc}
\usepackage{changepage}
\usepackage{algpseudocode}
\usepackage{booktabs}
\usepackage[dvipsnames]{xcolor}
\usepackage[ruled]{algorithm}
\usepackage{graphicx}
\usepackage[shortlabels]{enumitem}
\usepackage{fontawesome}
\usepackage{hyperref}
\usepackage[osf,sc]{mathpazo}
\usepackage[scaled=0.85]{helvet}
\usepackage{array}
\usepackage[kw]{pseudo}
\pseudoset{
  compact,
  label=\scriptsize\sffamily\color{gray}\arabic*,
  ctfont=\sffamily,
  ct-left=\quad\ctfont{[},
  ct-right=\ctfont{]},
}

\setlist[description]{parsep=0.1px}

\usepackage[shortlabels]{enumitem}

\newcommand{\field}{\mathbb{K}}

\newcommand{\affspace}[1]{{#1}^{n}}

\newcommand{\satu}[2]{(#1:#2^{\infty})}

\newcommand{\Oscar}{\href{https://oscar.computeralgebra.de/}{\sc Oscar}}
\newcommand{\msolve}{\href{https://msovle.lip6.fr/}{\tt msolve}}
\newcommand{\rad}{\operatorname{rad}}
\newcolumntype{?}{!{\vrule width 1pt}}

\usepackage{amsmath,amsfonts,amssymb,amsthm}
\newtheorem{definition}{Definition}
\numberwithin{definition}{section}
\newtheorem{theorem}[definition]{Theorem}
\newtheorem{corollary}[definition]{Corollary}
\newtheorem{proposition}[definition]{Proposition}
\newtheorem{lemma}[definition]{Lemma}
\theoremstyle{remark}
\newtheorem{remark}{Remark}[section]
\newtheorem{example}{Example}[section]

\usepackage[
citestyle=apa,
bibstyle=apa,
backend=biber,
isbn=false,
alldates=year,
sortcites=true,
doi=false,
eprint=true
]{biblatex}
\bibliography{direttissimo.bib}

\renewbibmacro*{doi+url}{%
  \iftoggle{bbx:doi}
  {\printfield{doi}%
    \iffieldundef{doi}{}{}}
  {}%
  \newunit\newblock
  \iftoggle{bbx:eprint}
  {\usebibmacro{eprint}%
    \iffieldundef{eprint}{}{}}
  {}%
  \newunit\newblock
  \iftoggle{bbx:url}
  {\usebibmacro{url+urldate}%
    \iffieldundef{url}{}{}}
  {}}

\AtEveryBibitem{\clearfield{month}}
\AtEveryBibitem{\clearfield{day}}
\DeclareNameAlias{sortname}{first-last}



\title{A \emph{Direttissimo} Algorithm for Equidimensional Decomposition}

\author{Christian Eder \thanks{RPTU Kaiserslautern-Landau, Germany}
  \and Pierre Lairez \thanks{Inria, Uni. Paris-Saclay, Palaiseau, France}
  \and Rafael Mohr \footnotemark[1] \footnotemark[3]
  \and Mohab Safey El Din \thanks{Sorbonne Uni., CNRS, Paris, France}
  }

\begin{document}

%\email{mohab.safey@lip6.fr}
\maketitle

\begin{abstract}
  We describe a recursive algorithm that decomposes an algebraic set
  into locally closed equidimensional sets, i.e. sets which each have
  irreducible components of the same dimension. At the core of this
  algorithm, we combine ideas from the theory of triangular sets,
  a.k.a. regular chains, with Gröbner bases to encode and work with
  locally closed algebraic sets. Equipped with this, our algorithm
  avoids projections of the algebraic sets that are decomposed and
  certain genericity assumptions frequently made when decomposing
  polynomial systems, such as assumptions about Noether position. This
  makes it produce fine decompositions on more structured systems
  where ensuring genericity assumptions often destroys the structure
  of the system at hand.  Practical experiments demonstrate its
  efficiency compared to state-of-the-art implementations.
\end{abstract}
\maketitle

\section{Introduction}
\label{sec:intro}

\paragraph{Problem statement}

Let $\field$ be an algebraically closed field, let $R = \field[x_1,\dotsc,x_n]$ be a polynomial
ring and let~$f_1,\dotsc,f_c \in R$ be a polynomial system generating an
ideal~$I\subseteq R$.  The zero set~$X$ of the
polynomials~$f_1,\dotsc,f_c$ in $\field$, decomposes uniquely as a union of irreducible
algebraic sets such that none of them contains
another.  These are the \emph{irreducible components} of~$X$ and
correspond to the \emph{minimal associated primes} of~$I$.  The
variety~$X$ is \emph{equidimensional} if all its irreducible
components have the same dimension.
% (Contrary \emph{unmixedness},
% this definition does not address embedded components or nonminimal
% associated primes.)
It is clear that~$X$ always admits a
decomposition~$X = Y_1 \cup \dotsb \cup Y_c$ where the~$Y_i$ are
equidimensional algebraic sets.

Suppose $X = \bigcup_{i=1}^s X_i$ is the decomposition of $X$ into
irreducible components. For $k=0,\dots,n$ define
$Y_k:= \bigcup_{i\text{ s.t. }\dim X_i=k}X_i$, then each $Y_k$ is
equidimensional so that $X = \bigcup_{k=0}^n Y_k$ is an
\textit{equidimensional decomposition} of $X$. Given
$f_1,\dots,f_c \in R$, we aim at computing such an equidimensional
decomposition of $X$.

It will be clear that our algorithms will only have to do arithmetic
over the subfield of $\field$ that the coefficients of $f_1,\dots,f_c$
lie in, the output will also be defined over the same subfield. In the
following we will work only over $\field$ for the purpose of simplicity
of presentation.

% where each component~$Y_i$ is to be described polynomial equations
% with coefficients in~$\field$.

% Returning algebraic representations of equidimensional
% $\field$-algeb\-raic sets has the advantage of providing {\em exact}
% representations of the encoded equidimensional components. This is of
% great interest for all applications where singularities are involved
% (see below) and for which numerical methods can be much less
% reliable. Usually, the field $\field$ is clear from the context. In
% this case, we omit it as a prefix.

This problem finds natural applications in singularity analysis of
sensor-based controllers of mechanism design \parencite[e.g.][and
references therein]{pascualescudero:hal-03070525,
  garciafontan:hal-03499974}, in algorithms of real algebraic geometry
\parencite[e.g.][]{AUBRY2002543, el2004properness} and real algebra
\parencite{SYZ18, SAFEYELDIN2021259} as well as automated theorem
proving and geometry \parencite[e.g.][]{wu2007mathematics,
  yang2001complete, yang1998automated, chen2013application}.

\paragraph{Prior Works}
The importance of this computational problem fostered a vast body of
literature often also as an intermediate step towards primary
decomposition of ideals or prime decomposition of
varieties. Algorithms for equidimensional decomposition of algebraic
sets can be classified along the data structures which they employ to
represent (equidimensional) algebraic sets.

There are two prominent strategies for equidimensional decomposition
using Gröbner bases frequently implemented in computer algebra
systems. The first one uses algebraic elimination techniques. It
combines the knowledge of the dimension of the ideal generated by the
input polynomials with the elimination theorem
\parencite[Theorem~3.1.2]{CoLiOS07} to compute a description of the
projection of the algebraic set under study on a well-suited affine
linear subspace to deduce how to split the corresponding ideal
\parencite{gianni1988, krick1991, caboara1997, decker1999,
  KALKBRENER1994365}.  The projection of the equidimensional component
of highest dimension (frequently called the \textit{equidimensional
  hull}) of the algebraic set under question will then be cut out by a
hypersurface whose defining polynomial has degree equal to the degree
of this equidimensional hull. As a consequence, such algorithms
have the disadvantage that they need to manipulate polynomials of
degree in the order of the Bezout bound of the input system. To
circumvent this drawback, another set of methods, called \emph{direct methods} has been introduced by
\textcite{eisenbud1992}.  They rely on homological algebra to reduce the
problem of equidimensional decomposition to the computation of
syzygies which are then used to split the polynomial ideal under
study, while avoiding projections.
These algorithms often provide an intermediate step towards
primary decomposition of ideals. For this problem, modular techniques
and dedicated algorithms for the case where $\field$ is a finite field
have been designed \parencite{NORO20041227, yokoyama2002prime,
  ishihara2022modular}.

% Finally, one should note that these algorithms often take an algebraic
% focus: They frequently define an ideal to be equidimensional when all
% of its associated primes have the same dimension rather than defining
% the notion of equidimensionality via algebraic sets alone.

Another body of work uses {\em lazy} representations of algebraic
sets. Frequently, the core idea is to exploit the fact that any
equidimensional algebraic set is locally almost everywhere a complete
intersection, i.e. an equidimensional algebraic set of codimension $c$
can be represented by the vanishing of $c$ polynomials on a dense
Zariski open subset of itself.  Hence, equidimensional algebraic sets
can be understood as the Zariski closures of locally closed sets
defined by polynomial equations and inequations. Taking this
perspective, one additionally enforces these $c$ polynomial equations
to form a {\em triangular set}. These have their origin in the Wu-Ritt
characteristic sets \parencite{ritt1950a, wu1986, chou1990, wang1993,
  Gallo1991}.  Triangularity is therein understood with respect to the
variables of the underlying polynomial ring (i.e. in a sense analogous
to the notion of triangular matrices in linear algebra). This
triangular structure naturally also yields the equations of the
algebraic set where the $c$ polynomials fail to define the algebraic
set at hand and thus triangular sets have a description of the
previously mentioned Zariski open subset attached to them in a natural
way. Because of their triangular structure they allow the reduction of
certain algorithmic challenges to a univariate problem. Of particular
importance, especially in the realm of equdimensional decomposition,
are certain special triangular sets called \textit{regular chains},
introduced by \textcite{kalkbrener1993, lu1994}. A regular chain
models an \textit{unmixed} dimensional ideal and has good algorithmic
properties with respect to the ideal it represents. A related
frequently implemented algorithm was also given by
\textcite{lazard1991}. Algorithms using regular chains are prominently
part of the computer algebra system Maple \parencite{chen2007a,
  chen2012}. We refer to \textcite{wang2001, hubert2003} for
introductions to the subject and to \textcite{aubry1999} for a
theoretical account as to how certain different notions of triangular
sets relate to each other.

It should be noted that, as for methods based on Gr\"obner bases
combined with algebraic elimination, these triangular encodings make
use of polynomials whose degrees, in the worst case, can be as high as
is the degree of the equidimensional components they do
encode. Nonetheless, algorithms based on triangular representations
can be quite well behaved compared to Gröbner basis techniques
especially on certain sparse polynomial systems.

Another data structure naturally encoding equidimensional algebraic
sets is that of a \textit{geometric resolution} developed by
\textcite{lecerf2000, lecerf2003}.  A geometric resolution is a
certain zero-dimensional parametrization of an algebraic set in
Noether position. In our setting, these zero-dimensional
parametrizations are used to encode generic points in the
equidimensional components of the algebraic set under study (the
numerical counterpart of this encoding is known as the notion of
witness sets \parencite{sommese2005}, a notion that will be utilized
in this paper as well). Under certain generically satisfied
assumptions on the input, these can be combined with \textit{straight
  line programs} to obtain the best known complexity bounds for
equidimensional decomposition. See also \parencite{jeronimo2002} for a
related approach.

To bypass the ``projection-degree'' problem, incremental
ap\-proa\-ches have been investigated in combination with Gr\"obner
bases algorithms.  Incremental means here that they feed the
decomposition algorithm with one input polynomial after another, in
the same way as \textcite{lazard1991} or \textcite{lecerf2000}, for
example, to identify when some polynomial is a zero divisor in the
ring of polynomials quotiented by the ideal generated by the previous
polynomials.  \Textcite{moroz2008} combines Gr\"obner bases
computations with representations of equidimensional algebraic sets by
means of locally closed sets. In a previous work, we also investigated
this approach by exploiting properties of signature-based Gr\"obner
bases algorithms to enhance the detection and exploitation of zero
divisors and compute the so-called nondegenerate locus of a polynomial
system \parencite{eder2022}.

\paragraph{This Work}
In this work, we again take the incremental approach previously mentioned.
As in the other incremental algorithms, the foundation of our
algorithm is a decomposition algorithm to, given an equidimensional
algebraic set $X$ and some $f\in R$, determine the \textit{locus of
  proper intersection} of $f$ on $X$, i.e. the set of points
$p\in k^n$ such that $X\cap V(f)$ has dimension one less than $X$. This is
then used to iterate over the input equations $f_1,\dots,f_r$. More
precisely, one starts by decomposing $V(f_1,f_2)$ then uses the output
to decompose $V(f_1,f_2,f_3)$ and so on.

In contrast to a lot of other algorithms for equidimensional
decomposition based on Gröbner bases we borrow from the theory of
triangular sets and work with locally closed sets instead of
polynomial ideals simlarly to \textcite{moroz2008}. In the iterative
approach outlined above this turns out to have two benefits.

First, it naturally removes from the output sets of our iterative
algorithm certain embedded components that appear during the
decomposition. To illustrate this consider the following example:
\begin{example}
  Let $R:=\mathbb{Q}[x,y,z]$, $X:=V(xy), f:=xz$. To decompose $X\cap V(f)$ into
  equidimensional components one may start by decomposing $X = V(x) \cup V(y)$.
  Then one intersects these two components with $V(f)$ to obtain the
  equidimensional decomposition $X\cap V(f) = V(x) \cup V(y,xz)$. The latter set
  has the irreducible component $V(y,x)$ which is embedded in $V(x)$. If one
  instead splits into a {\em disjoint} union $X = V(x) \cup \left[ V(y)\setminus
    V(x) \right]$ and again intersects both components with $V(f)$ one obtains
  $X = V(x) \cup \left(V(y,z)\setminus V(x) \right)$ and the latter component no
  longer has the irreducible component $V(y,x)$.
\end{example}

Second, an iterative equidimensional decomposition algorithm may
produce redundant components, which, if they are not deduplicated, may
yield an exponential blow-up in the number of components: if one has
decomposed $X = \bigcup_i X_i$ with the $X_i$ sharing a large number of
irreducible components then decomposing each $X_i\cap V(f)$ to obtain a
decomposition of $X\cap V(f)$ results in an even more redundant
decomposition.  Because we use locally closed sets to model our
equidimensional sets we are enabled to enforce that every time we
decompose a locally closed set the resulting output sets be pairwise
set-theoretically disjoint. Our experiments indicate that this seems
to enforce a sufficiently strong irredundancy between our components
to avoid an exponential blow-up in the number of components.

In this paper we provide two methods to work with the locally closed
sets appearing in our algorithm: One method models them ``naively'' in
the sense that we encode them by storing their defining equations and
inequations and use Gröbner bases of their associated ideals to work
with them algorithmically. The other method tries to avoid having to
know a Gröbner basis for the ideal associated to a locally closed set
as much as possible by storing instead a Gröbner basis for a
\emph{witness set} of the locally closed set in question. Using
Gröbner bases here with the graded reverse lexicographical ordering
has the effect that, compared to algorithms using triangular sets, we
are able to
\begin{itemize}
\item avoid computing projections of the algebraic sets to be
  decomposed and certain frequently made genericity assumptions such
  as ideals being in Noether position;
\item obtain desciptions of these sets with lower degree polynomials.
\end{itemize}

Borrowing further from the theory of triangular sets we also adopt the
heuristic that it is a good idea to decompose given algebraic sets as
often and as finely as possible when working with them. This
philosophy is baked into the recursive structure of our algorithms
which exists so as to decompose a given locally closed set as much as
possible given generating sets for certain saturation ideals.

We implemented our algorithm in the computer algebra system \Oscar{}
\parencite{Oscar2023} using its interface to the library \msolve{}
\parencite{berthomieu2021} for all necessary Gröbner basis
computations. Experimental results indicate that our algorithm is able
to tackle polynomial systems which are out of reach of state-of-the
art implementations of algorithms for equidimensional decomposition
which are available in leading computer algebra systems.

% \paragraph{Structure of the paper}

% In Section~\ref{sec:algop}, we give the intuitive idea behind our algorithms. We
% then define the notion of an \textit{affine cell} as particular locally closed
% sets, all of the components in our equidimensional decomposition algorithm will
% be affine cells. We then describe our algorithms in pseudo-code. This
% pseudo-code is agnostic of the concrete data structure representing affine
% cells, so that Section~\ref{sec:algop} is focused purely on the geometric
% content of our algorithms. % In Section~\ref{sec:basedat}, we then repeat some
% % necessary statements from commutative algebra in order to prove correctness and
% % termination of our algorithm in Section~\ref{sec:corterm}.
% Correctness and termination of our algorithm are proved in
% Section~\ref{sec:corterm}. We then go into more detail as to how we implemented
% our algorithm in Section~\ref{sec:implem}, describing a probabilistic
% optimization as well as a more optimal ordering of certain operations our
% algorithm performs. We also give some experimental results in
% Section~\ref{sec:implem} where we compare timings of running our algorithm on
% certain example polynomial systems to other algorithms implemented in different
% computer algebra systems.

\section{Algorithms}
\label{sec:algop}

\subsection{Principles}

To illustrate the basic principles behind our equidimensional
decomposition algorithm, consider an equidimensional variety~$X$ in
the affine space~$\affspace{\field}$.  Let~$f \in R$.  The
variety~$X$ is partitioned into:
\begin{enumerate}
\item Points~$p$ where~$f$ is a non zero divisor locally at~$p$ (that
  is in the ring~$R_p/I(X) R_p$). The polynomial $f$ takes nonzero
  values in any open neighborhood of~$p$ in~$X$.\\
  This defines an open
  subset~$X_\text{proper}$ of~$X$.
\item Points~$p$ contained in an irreducible component of~$X$ on which~$f$
  vanishes identically. \\
  This defines a closed subset~$X_\text{improper}$ of~$X$.
\end{enumerate}
It is clear that~$X = X_\text{proper} \sqcup X_\text{improper}$
(where~$\sqcup$ denotes a disjoint union) and
that~$X_\text{improper} \subseteq V(f)$, so that
\[ X\cap V(f) = \left( X_\text{proper} \cap V(f) \right) \sqcup
  X_\text{improper}. \] By construction, the
$ X_\text{proper} \cap V(f)$ is a \emph{proper} intersection: it is
equidimensional of dimension~$\dim X -1$, or empty. As a union of
irreducible components of~$X$, the closed set $X_\text{improper}$ is
equidimensional, with the same dimension as the one of~$X$, unless it
is empty. So we obtain an equidimensional decomposition
of~$X\cap V(f)$. Given defining equations for $X$, this process can be
applied iteratively to obtain an equidimensional decomposition of any
affine algebraic variety.

In our algorithm we apply the above idea without directly computing
$X_{{text}{proper}}$ and $X_{\text{improper}}$. Let $I(X)\subset R$ be an
ideal such that $V(I(X)) = X$. Further, we denote by $\satu{I(X)}{f}$
the saturation ideal of $I(X)$ by $f$. Recall that $V(\satu{I(X)}{f})$
is the Zariski closure of $X \setminus V(f)$
\parencite[Theorem~4.4.10]{CoLiOS07}.

We look for an element~$g \in \satu{I(X)}{f} \setminus I(X)$. If there is none,
this implies that~$X_\text{improper} = \varnothing$ so~$X \cap V(f)$ is
equidimensional. If there is such a $g$, then we consider the
following partition of~$X$:
\begin{enumerate}
\item the closed locus~$X_1$ of points~$p$ where~$g$ has nonzero
  values in any neighborhood of~$p$ in~$X$;
\item the open locus~$X_2$ of points~$p$ where~$g$ is zero in some
  neighborhood of~$p$ in~$X$.
\end{enumerate}
These two sets are equidimensional.  By construction, $fg$ vanishes
identically on~$X$, so~$X_1 \subseteq X_\text{improper}$ and this gives the
following decomposition of~$X$:
\begin{equation}
  X = X_1 \sqcup \left( X_2\cap V(f) \right).\label{eq:2}
\end{equation}
The ideal of~$X_1$ is given by~$\satu{I(X)}{g}$.  The term
$X_2\cap V(f)$ may not be equidimensional but we may apply the above idea
recursively: We again split $X_2$ along an element in
$\satu{I(X_2)}{f}\setminus I(X_2)$ if it exists.  This leads to Algorithm
\emph{split}.

The set~$X_2$ is not closed, this raises the need to deal not only
with closed sets of the affine space, but more generally locally
closed sets.  We do so by partitioning them into special locally
closed sets, more precisely into closed sets in the complement of a
hypersurface in the affine space, which we call \emph{affine cells}.
Concretely, suppose that $I(X_1) = \satu{I(X)}{g}$ is given by a
finite generating set $H\sqcup \{h\}\subset R$. We then recursively decompose
$X_2 = X\setminus V(H\sqcup \{h\})$ via
\[ X_2 = X \setminus V(H \cup \left\{ h \right\}) = ( X\setminus V(h)) \sqcup \big( (X\setminus V(H))
  \cap V(h) \big). \] The intersection with~$V(h)$ is computed with
\emph{split} to ensure equidimensionality.  Algorithm~\emph{remove}
below performs these operations.  Findally we obtain an
equidimensional decomposition algorithm following an incremental
strategy by repeated application of \emph{split}, see
Algorithm~\emph{equidim}.

The primitive operations we use to manipulate affine cells are
presented next, while the proof of correctness and termination of the
algorithms are in Section~\ref{sec:corterm}.

\subsection{Primitives}
\label{sec:primitives}

\begin{definition}
  \label{def:hypcomp}
  An \emph{affine cell}~$X$ is a locally closed set
  of~$\field^n$ of the form~$Z\setminus V(g)$ where~$Z$ is an
  algebraic set and~$g\in R$. An affine cell~$X$ is
  \emph{equidimensional} if all the irreducible components of the
  Zariski closure~$\overline X$ have the same dimension.
\end{definition}

Regardless of the mode of representation of an affine cells, we
assume that we can perform the following operations on any affine
cell~$X$:
\begin{enumerate}[(1)]
  \item\label{it:inter} Given~$f\in R$, compute the affine cell~$X\cap V(f)$;
  \item\label{it:complement} Given~$f\in R$, compute the affine cell~$X \setminus V(f)$;
\end{enumerate}

As often in effective algebraic geometry,
algebraic sets are defined by ideals that are not always radical
so our affine cells come with a distinguished ideal~$I(X)\subseteq R$
such that~$\overline X = V(I(X))$.
The radical of~$I(X)$ is denoted~$\rad I(X)$.
We assume that operations~\ref{it:inter} and~\ref{it:complement}
satisfy $I(X) + \langle f\rangle \subseteq I( X\cap V(f) )$ and~$I(X) \subseteq I(X \setminus V(f))$.
We assume further that we can perform the following operations on any affine cell~$X$:
\begin{enumerate}[(1), resume]
  \item Given~$f\in R$, decide if~$f\in I(X)$;
  \item Compute a basis of $I(X)$, denoted~\fn{basis}(X).
\end{enumerate}

For example, we may represent an affine cell~$X$ by a pair~$(F, g)$,
where~$F$ is a Gröbner basis of~$I(X)$, for some monomial ordering,
and~$g$ a polynomial such that~$X = \overline{X}\setminus V(g)$
\parencite[see][for an introduction to Gröbner bases]{becker1993}.  We
denote~$X = V(F; g)$.  For a set $F\subseteq R$ and an element
$g\in R$, let~$\fn{sat}(F, g)$ denote a Gröbner basis of the saturation
ideal~$\satu{\langle F \rangle}{g}$.  Recall that
\[ \satu{I}{g} \stackrel{\text{def}}= \left\{ f\in R \ \middle|\  \exists k\in \mathbb{N}, fg^k \in I \right\}.\]
Using these primitive \fn{sat}, we can perform all the four operations above:
\begin{enumerate}[(1)]
  \item $V(F; g) \cap V(f) = V( \fn{sat}(F \cup \left\{ f \right\}, g); g)$;
  \item $V(F; g) \setminus V(f) = V( \fn{sat}(F, f); fg)$;
  \item $f \in I(X)$ if and only if the normal form of~$f$ w.r.t.~$F$ is zero;
  \item $\fn{basis}(V(F; g)) = F$.
\end{enumerate}

\begin{remark}
  In Section \ref{sec:implem} we explain how to perform the above
  primitive operations on an affine cell $X$ using a notion called
  \textit{witness sets}, introduced for the purpose of equidimensional
  decomposition by \textcite{lecerf2003} under the name lifting
  fibers.  This leads to a lazier representation of $X$, one where a
  Gröbner basis for $I(X)$ is not always required.
\end{remark}

\begin{algorithm}[H]
  \caption{Equidimensional decompositions}
  \label{alg:split}
  \raggedright

  \begin{description}
    \item[Input] An equidimensional affine cell~$X$, an element~$f\in R$
    \item[Output] A partition of~$X\cap V(f)$ into equidimensional affine cells
  \end{description}

  \begin{pseudo}
    function \fn{split}(X, f)\\+
    $G \gets \fn{basis}( X\setminus V(f))$\\
    if $G \subseteq I(X)$ \label{line:testproper} \ct{can be replaced by~$G\subseteq \rad I(X)$}\\+
    return $\left\{ X\cap V(f) \right\}$ \label{line:proper}\\-
    else\\+
    $g\gets $ \tn{any element of~$G \setminus I(X)$} \label{line:pickg}\\
    $H \gets \fn{basis}(X \setminus V(g))$\\
    $\mathcal{D} \gets\left\{ X\cap V(H) \right\}$ \label{line:improper}\\
    for $Y \in \fn{remove}(X \cap V(g), H)$ \label{line:addg}\\+
    $\mathcal{D} \gets \mathcal{D} \cup \fn{split}(Y, f)$\label{line:reccall}\\-
    end\\
    return $\mathcal{D}$\\-
    end\\-
    end
  \end{pseudo}

  \bigskip
  \begin{description}[labelsep=1em]
    \item[Input] An  affine cell~$X$, a finite set~$H \subset R$
    \item[Precondition] $X\setminus V(H)$ is equidimensional
    \item[Output] A partition of~$X\setminus V(H)$ into equidimensional affine cells
  \end{description}

  \begin{pseudo}
    function \fn{remove}(X, H)\\+
    if $H = \varnothing$\\+
    return $\varnothing$\\-\
    else\\+
    $h\gets $ \tn{any element of~$H$}\\
    $\mathcal{D} \gets \left\{ X \setminus V(h) \right\}$\\
    for $Y \in \fn{remove}(X, H\setminus \left\{ h \right\})$\\+
    $\mathcal{D} \gets \mathcal{D} \cup \fn{split}(Y, h)$\\-
    end\\
    return $\mathcal{D}$\\-
    end\\-
    end
  \end{pseudo}

  \bigskip
  \begin{description}[labelsep=1em]
    \item[Input] a finite set $F \subseteq R$
    \item[Output] A partition of~$V(F)$ into equidimensional affine cells
  \end{description}
  \begin{pseudo}
    function \fn{equidim}(F)\\+
    $\mathcal{D} \gets \left\{ V(\varnothing; 1) \right\}$ \ct{the full affine space}\\
    for $f$ in $F$\\+
    $\mathcal{D} \gets \bigcup_{X\in \mathcal{D}} \fn{split}(X, f)$\\-
    end\\
    return $\mathcal{D}$\\-
    end
  \end{pseudo}
\end{algorithm}

\begin{example}
  \label{exp:split}
  To illustrate Algorithm \emph{split} we spell out how it behaves on
  the input $X := V(xy,zw)$ and $f:=xz$. Using the notation of
  Algorithm \emph{split} we find $G = \{y,w\}$. This is not contained
  in $I(X)$, so we may choose $g := y$ in line 6 of
  Algorithm \emph{split}. Then we find $H = \{x,zw\}$. Note that
  $X\setminus V(zw) = \emptyset$ and so Algorithm \emph{split} returns
  \begin{align*}
    &X\cap V(H) \text{ and }\text{\emph{split}}(\text{\emph{remove}}(X,H),f)\\
    =\; &V(zw,x) \text{ and } \text{\emph{split}}(V(y,zw)\setminus V(x), xz).
  \end{align*}
  This second call to Algorithm \emph{split} finds $G = \{y,w\}$,
  again this set is not contained in $I(V(y,zw)\setminus V(x))$, and so we can choose
  $g := w$ in line 6. Then we find $H = \{z\}$ which this time yields
  \begin{align*}
    \text{\emph{split}}(V(y,zw)\setminus V(x), xz) = &V(y,z)\setminus V(x) \\
    &\text{and } \text{\emph{split}}(V(y,w)\setminus V(xz),xz)
  \end{align*}
  The last call to split simply finds the empty set and so all
  in all we have obtained the decomposition
  \begin{align*}
    V(xy,zw,xz) = V(x,zw) \cup V(y,z)\setminus V(x).
  \end{align*}
\end{example}

\begin{remark}
  Example \ref{exp:split} illustrates the fact that Algorithm
  \emph{split} may split an algebraic set even if it is
  equidimensional. Heuristically, the finer the intermediate
  decomposition in Algorithm \emph{equidim} is, the computationally
  easier subsequent steps will be.
\end{remark}

\subsection{Correctness and Termination}
\label{sec:corterm}

When computing an interection of an equidimensional affine cell~$X$
with a hypersurface~$V(f)$, we distinguish two cases, depending on whether $V(f)$ intersects~$X$ properly or not.
Lemma~\ref{lem:regint} deals with the first case, while Lemma~\ref{lem:irregular-step} deals with the second.

% \begin{lemma}\label{lem:one-dec-step}
%   Let~$X$ be a nonempty equidimensional affine cell. Let~$g\in R$ such that~$g \not\in \operatorname{rad} I(X)$ and let~$J\subseteq \satu{I(X)}{g}$ be an ideal.
%   Let~$Y = X \setminus V(J)$.
%   \begin{enumerate}[(i)]
%     \item \label{it:equidim} $Y$ is empty or equimensional of same dimension as~$X$;
%     \item \label{it:increase} $I(X) \subset \satu{I(X)}{I(Y)}$ (strict inclusion);
%     \item \label{it:partition} $X \cap V(J) = Y \sqcup \left( (X\setminus Y) \cap V(J) \right)$.
%   \end{enumerate}
% \end{lemma}

% \begin{proof}
%   $Y$ is a Zariski open set in~$X$, so it inherits the equidimensionality and the dimension of~$X$,
%   this gives~\ref{it:equidim}.

%   It is clear that~$I(X) \subseteq \satu{I(X)}{I(Y)}$.
%   To prove that the inclusion is not an equality, it is enough to check that
%   the radical of the right-hand side is not included in that of the left-hand side.
%   Indeed,
%   $\satu{I(Y)}{g} = (I(Y) : g^m)$ for some~$m \geq 0$,
%   so~$g^m \in \left( I(Y) : (I(Y) : g^m) \right) \subseteq \satu{I(X)}{I(Y)}$.
%   By hypothesis, $g^m$ is not in the radical of~$I(X)$.
%   This proves~\ref{it:increase}.

%   Finally, since~$J g^k \subseteq I(X)$, for some~$k \geq 0$, we may partition the points of~$X$ where~$f =0$
%   into the set of points of~$X$ not in~$V(g)$ (this is~$Y$)
%   and the set of points of~$X$ in both~$V(g)$ and~$V(J)$ (this is~$(X\setminus Y) \cap V(J)$).
% \end{proof}

\begin{lemma}
  \label{lem:regint}
  Let~$X$ be an equidimensional affine cell.
  Let~$f \in R$, such that $\satu{I(X)}{f} \subseteq \rad I(X)$.
  Then~$X \cap V(f)$ is empty or equidimensional with dimension~$\dim X - 1$.
\end{lemma}
\begin{proof}
  Let~$I = I(X)$.  We may assume that~$I$ is radical: the
  assumption $\satu{I}{f} \subseteq \rad I$ also
  implies~$\satu{\rad(I)}{f} \subseteq \rad I$.
  If~$a \in \satu{\rad(I)}{f}$, then~$a f^r \in \rad I$, for
  some~$r \geq 0$, and so~$(a f^r)^s \in I$, for some~$s \geq 0$. In
  particular, $a^s \in \satu{I}{f} \subseteq \rad I$.
  So~$a\in \rad I$.  Suppose that $X \cap V(f)$ is not empty. By Krull's
  principal ideal theorem any minimal prime over
  $I + \langle f \rangle$ has codimension at most
  $\operatorname{codim}I + 1$.  The
  condition~$\satu{I}{f} \subseteq \rad I$ means geometrically
  that~$X\subseteq \overline{X\setminus V(f)}$, so that $f$ has nonzero values in the
  neighborhood of any point in~$X$.  So~$f$ is a not a zero divisor
  in~$R/I$.  In particular, there is a regular sequence of length
  $\operatorname{codim}I + 1$ in $I + \langle f \rangle$.  Since the polynomial
  ring $R$ is Cohen-Macaulay it follows that every minimal prime over
  $I + \langle f \rangle$ has at least codimension $\operatorname{codim} I + 1$.
\end{proof}

\begin{lemma}\label{lem:irregular-step}
  Let~$X$ be an equidimensional affine cell.
  Let~$f \in R$, let~$g \in \satu{I(X)}{f}$
  and let~$I_g = \satu{I(X)}{g}$.
  Let~$X_1 = X\cap V(I_g)$ and~$X_2 = (X \cap V(g)) \setminus V(I_g)$.
  Then:
  \begin{enumerate}[(i)]
    %\item \label{it:partition} and $X \cap V(f) = (X \cap V(J)) \sqcup \left( (X\setminus V(J)) \cap V(f) \right)$.
    \item \label{it:part0} $X = X_1\sqcup X_2$;
    \item \label{it:partition} $X \cap V(f) = X_1 \sqcup \left( X_2 \cap V(f) \right)$ ;
    \item \label{it:equidim1} $X_1$ is empty or equidimensional with~$\dim X_1 = \dim X$;
    \item \label{it:equidim2} $X_2$ is empty or equidimensional with~$\dim X_2 = \dim X$;
  \end{enumerate}
\end{lemma}

\begin{proof}
  Obviously $X = X_1 \sqcup (X \setminus V(I_g))$.
  As a set, $X_1$ is the union of the components of~$X$ on which~$g$ is not identically zero.
  In particular~$X \setminus V(I_g)$ is the set of points of~$X$ in a neighborhood of which~$g$ is identically zero.
  Therefore~$X \setminus V(I_g) \subseteq V(g)$, so we obtain
  \[ X \setminus V(I_g) = (X \cap V(g)) \setminus V(I_g), \]
  which gives~\ref{it:part0}.

  Next, we have~$I(X_1) = I(X) + I_g = \satu{I(X)}{g}$.
  Moreover~$f \in \operatorname{rad} I(X_1)$.
  Indeed, $g f^k \in I(X)$ for some~$k \geq 0$, by definition of~$g$,
  and therefore~$f \in \operatorname{rad} \left( I(X) : g \right) \subseteq \operatorname{rad}\satu{I(X)}{g}$.
  So~$X_1 \subseteq V(f)$.
  It follows that~$X \cap V(f) = X_1 \sqcup \left( X_2 \cap V(f) \right)$.
  This proves~\ref{it:partition}.

  Since~$X$ is equimensional, it follows that~$X_1$ (as a union of components of~$X$) is also equidimensional
  of same dimension, unless it is empty.
  This proves~\ref{it:equidim1}.
  As for~$X_2$, it is open in~$X$, so it inherits the equidimensionality and the dimension of~$X$, unless it is empty.
  This proves~\ref{it:equidim2}.
\end{proof}

We now prove correctness and termination of Algorithms \emph{split} and \emph{remove}
with a mutual induction.
On line~\ref{line:testproper}, the test~$G\subseteq I(X)$ can be replaced by~$G\subseteq \rad I(X)$,
or any condition which holds when~$G\subseteq I(X)$ and doee not hold when~$G\not\subseteq \rad I(X)$,
this does not affect correctness or termination. We will use this variant in Section~\ref{sec:implem}.

\begin{theorem}\label{lem:term-correc-technical}
  For any affine cell~$X$:
  \begin{enumerate}[(i)]
    \item \label{it:split} If~$X$ is equidimensional, then for any~$f\in R$, the procedure~\emph{split} terminates on input~$X$ and~$f$
    and outputs a partition of\/~$X\cap V(f)$ into equidimensional affine cells~$Y$ with~$I(X) \subseteq I(Y)$.
    \item \label{it:remove} For any finite set~$H\subset R$ such that~$X \setminus V(H)$ is equidimensional,
    the procedure~\emph{remove} terminates on input~$X$ and~$H$ and outputs a partition of\/~$X\cap V(H)$ into equidimensional affine cells~$Y$ with~$I(X) \subseteq I(Y)$;
  \end{enumerate}
\end{theorem}

\begin{proof}
  We proceed by Noetherian induction on~$I(X)$ and assume the statement holds for any affine cell $X'$ with~$I(X) \subsetneq  I(X')$.

  We begin with \fn{split}.
  Let~$f \in R$ and let~$I_f = \satu{I(X)}{f}$.
  If~$I_f \subseteq I(X)$,
  then Lemma~\ref{lem:regint} applies and~$X\cap V(f)$ is equidimensional.
  So $\fn{split}(X, f)$ terminates and is correct in this case.

  Assume now that there is some~$g\in I_f \setminus I(X)$.
  Let~$I_g= \satu{I}{g}$.
  Lemma~\ref{lem:irregular-step} applies: an equidimensional decomposition of~$X\cap V(f)$
  is given by~$X\cap V(I_g)$ and an equidimensional decomposition of~$\left( (X\cap V(g))\setminus V(I_g) \right)\cap V(f)$.
  Moreover~$(X\cap V(g)) \setminus V(I_g)$ is equidimensional.
  Since~$g\not\in I(X)$, we have~$I(X) \subsetneq I(X\cap V(g))$
  so \fn{remove}(X\cap V(g), H)  (using the notations of Algorithm \emph{split}, where~$H$ is a generating set of~$I_g$)
  is correct and terminates, by induction hypothesis. Moreover, it outputs affine cells~$Y$ such that~$I(X) \subsetneq I(X \cap V(g)) \subseteq I(Y)$.
  So the recursive calls \fn{split}(Y, f) are correct and terminate.

  As for \fn{remove},
  let~$H \subset R$ finite such that~$X \setminus V(H)$ is equidimensional.
  If~$H = \varnothing$, then~\ref{it:remove} holds trivially.
  As for the case~$H\neq \varnothing$, let~$ h\in H$ and~$H' = H\setminus h$.
  Since~$V(H) = V(h) \cap V(H')$, we have
  \begin{equation}
    X \setminus V(H) = (X \setminus V(h)) \sqcup \left( (X \setminus V(H')) \cap V(h) \right).\label{eq:1}
  \end{equation}
  The set $X\setminus V(h)$ and $X\setminus V(H')$ are open in~$X\setminus V(H)$ so equidimensional (or empty).
  By induction on the cardinal of~$H$, we assume that \fn{remove}(X, H') is a partition of~$X \setminus V(H')$ into equidimensional affine cells,
  and that every cell $Y$ of this partition satisfies~$I(X) \subseteq I(Y)$.
  By~\ref{it:split}, the calls
  \fn{split}( Y, h)  terminates and yield a partition of
  $(X \setminus V(H')) \cap V(h)$ into cells~$Y$ with~$I(X) \subseteq I(Y)$.
  Moreover the affine cell~$Y = X\setminus V(h)$ also satisfies~$I(X) \subseteq I(Y)$.
  By~\eqref{eq:1}, \fn{remove}(X, H) terminates too and is a partition of~$X \setminus V(H)$
  into cells~$Y$ with~$I(X) \subset I(Y)$.
\end{proof}

\begin{corollary}
  Algorithm \emph{equidim} is correct and terminates.
\end{corollary}

\section{Implementation and experimental results}
\label{sec:implem}

\subsection{Implementation Details}
\label{sec:algdet}

In this section we give some implementation details and alternatives.
In particular, we show a lazier data structure for affine cells
which is able to delay some Gröbner basis computations
at the cost of a Monte Carlo randomization.
We have implemented both the method described
in Section \ref{sec:algop} and the method described in this section.

For either method, we will need an algorithm that, given generators
for an ideal $I$ and an element $f\in R$, computes generators for the
saturation $\satu{I}{p}$. Even for our lazy representation, this will
still sometimes be needed to compute a Gröbner basis for the ideal
$I(X)$, where $X$ is an affine cell. In the probabilistic setting,
some saturations will be replaced by saturations of zero dimensional
ideals.

In our implementation we chose the standard method of performing
saturations using Gröbner bases.
To compute generators for
$\satu{I}{p}$, fix a monomial order $\leq$ on $R[t]$ for a new variable
$t$ such that $\leq$ eliminates $t$. Compute a Gröbner basis $G$ for the
ideal $I + \langle tp - 1\rangle \subset R[t]$ w.r.t $\leq$. Then the elements in
$G$ that do not contain the variable $t$ give a Gröbner basis of
$\satu{I}{p}$ by the elimination theorem.
% TODO: should we cite these?
Other saturation methods also exist such as the methods
presented in \textcite{eder2022} or \textcite{berthomieu2022}.

Randomization relies on intersecting with random linear subspaces of
appropriate dimension to reduce to the zero-dimensional case. This
idea is well known in symbolic computation \parencite{lecerf2003} and
numerical algebraic geometry \parencite[e.g.]{bates2013} wherein these
intersections of algebraic sets with random suitable random linear
subspaces are known under the name \emph{witness sets}.

\begin{proposition}\label{prop:proba}
  Let $X\subseteq \affspace{\field}$ be  an equidimensional affine cell of dimension
  $d$ and let $f\in R$.
  Then, for a
  generic linear subspace $L\subset \affspace{\field}$ of codimension $d$ the
  following statements hold:
  \begin{enumerate}
  \item $f\in \rad{I(X)}$ if and only if $f \in \rad{I(X\cap L)}$.
  \item $I(X\setminus V(f)) \subseteq \rad{I(X)}$ if and only if $X \cap L \cap V(f) = \varnothing$.
  \end{enumerate}
\end{proposition}

\begin{proof}
  We always have~$\rad I(X) \subseteq \rad I(X \cap L)$.
  Conversely,
  assume that~$f\not\in \rad I(X)$.
  Let~$U = \left\{ p\in X\ \middle|\  f(p) \neq 0 \right\}$/
  It is an open subset of~$X$ and it is non empty, by hypothesis.
  Since~$X$ is equidimensional, $U$ has dimension~$d$
  and the intersection~$U\cap L$ is nonempty (because~$L$ is generic).
  Therefore~$f$ is nonzero on a nonempty subset of~$X\cap L$.
  In particular, $f \not \in \rad I(X\cap L)$.
  This proves the first point.

  For the second point note that $I(X\setminus V(f)) \subseteq \rad{I(X)}$
  if and only if~$X$ and~$V(f)$ intersect properly,
  that is~$X\cap V(f)$ is equidimensional of dimension~$d-1$.
  The intersection of~$X\cap V(f)$ with the codimension~$d$ generic space~$L$ is empty if and only if the dimension of~$X\cap V(f)$ is less than~$d$.
  The proves the second point.
\end{proof}

In this setting, we represent an equidimensional affine cell~$X$ by a
triple~$(F, G, W, d)$, where~$F$, $G$ and~$W$ are subsets of~$R$
and~$d$ is an integer such that~$\dim X = d$,
$X = V(F) \setminus V(\prod_{g\in G} G)$ and~$W$ (stands for \emph{witness set}) is
a Gröbner basis of~$I(X \cap L)$ for some generic linear subspace
space~$L$ of~$\field^n$ of codimension~$d$.  We
denote~$X = V(F; G, W, d)$.  In practice, $L$ will only be random and
sufficient genericity will only hold with high probability (assuming
that~$\field$ has enough elements).  Given only $F$, $G$ and~$d$, we
can compute a suitable set~$W$ by choosing a set~$J\subseteq R$ of~$d$ random
linear forms and computing a Gröbner basis
of~$(((F : g_1^\infty) : \dotsb): g_r^\infty)$,
where~$G = \left\{ g_1,\dotsc,g_r \right\}$.  This procedure is
denoted \fn{witness}(F, G, d).

The four primitive operations are perfomed as follows.
For the intersection operation, we need some additional knowledge on the expected dimension of the output.
Let~$X = V(F; G, W, d)$ be an equidimensional cell.
\begin{enumerate}[(1)]
  \item \emph{[Proper intersection]} Given $f\in R$ such that~$X$ intersects~$V(f)$ properly,
  \[ X \cap V(f) = V\big(F' ; G, \fn{witness}(F', G, d-1), d-1\big), \]
  with~$F' = F \cup \left\{ f \right\}$;
  \item[(1')] \emph{[Purely improper intersection]} Given~$H\subset R$ such that~$X \cap V(H)$ is a union of components of~$X$,
  \[ X \cap V(H) = V\big(F \cup H ; G, \fn{gb}(W \cup H), d \big), \]
  where \fn{gb}(W\cup H) denotes a Gröbner basis of the ideal generated by~$X\cup H$;
  \item for~$f\in R$, $X\setminus V(f) = V(F; G\cup \left\{ f \right\}, \fn{sat}(H, f), d)$;
  \item $f\in \rad I(X)$ if and only if~$1 \in (W : f^\infty)$;
  \item $I(X)$ is computed by saturating~$\langle F\rangle$ successively by all the elements of~$G$.
\end{enumerate}

In the decomposition algorithm, we always know \emph{a priori}
the kind of each intersection.
The intersection on line~\ref{line:proper} of \fn{split} is proper,
the intersection on line~\ref{line:improper} is purely improper.
The one on line~\ref{line:addg} is more subtle. Indeed, the decomposition algorithm may produce here a nonequidimensional cell
when considering~$X\cap V(g)$.
With the notations of this algorithm, the cell~$X' = X\cap V(g)$ is only equidimensional outside of~$V(H)$ (of dimension~$\dim X$).
This nonequidimensional cell will go through only one operation among the four primitives: $X'\setminus V(h)$ for some~$h\in H$.
This operation restores equidimensionality.
So we can mostly ignore this issue and compute the intersection $X\cap V(g)$ as a purely improper intersection, pretending that~$X\cap V(g)$ is equidimensional.

In addition we obtain
a fifth operation: a probabilistic algorithm to
check if $X\cap V(f)$ is empty or equidimensional of dimension one less
than $X$ (or, equivalently $\satu{I(X)}{f} \subseteq \rad I(X)$).
This is given by Algorithm \emph{isProper}.
Equipped with this algorithm, we can replace the if-condition in line~\ref{line:testproper} of Algorithm \emph{split} with
$\fn{isProper}(X,f)$. Only if this is not satisfied we proceed to
compute a Gröbner basis for $I(X\setminus V(f))$.

\begin{algorithm}[]
  \caption{Proper intersection check}
  \label{alg:isreg}
  \raggedright

  % \begin{description}[labelsep=1em]
  % \item[Input] Finite sets of polynomials $F,G\subset R$ determining
  %   the affine cell $X := V(F) \setminus V(\prod_{g\in G}g)$, a target dimension $d$
  % \item[Output] A Gröbner basis of a zero-dimensional ideal
  %   whose zero locus is contained in the components of $X$
  %   of dimension $d$
  % \end{description}

  % \begin{pseudo}
  %   function \fn{WitnessSet}(F,G,d)\\+
  % $L \gets $ a set of $d$ random linear forms in $R$\\
  % $G \gets $ a Gröbner basis of $\satu{\langle F\cup L \rangle}{\prod_{g\in G}g}$\\
  % return $G$\\-
  % end
  % \end{pseudo}

  \begin{description}[labelsep=1em]
    \item[Input] An equidimensional affine cell~$X$, an element~$f\in R$
    \item[Output] \kw{true} if $X\cap V(f)$ is a proper intersection, \kw{false} otherwise
  \end{description}

  \begin{pseudo}
    function \fn{isProper}(X, f)\\+
    $W \gets $ \tn{the withness set of~$X$}\\
    $W' \gets$ \tn{a Gröbner basis of $\satu{\langle W\rangle}{f}$}\\
    return $1 \in W'$
  \end{pseudo}
\end{algorithm}

Lastly we want to note the following: In Algorithm \emph{split}, on
input $X$ and $f$, we may have to compute
$G := \fn{basis}(X\setminus V(f))$ but we use only one element in $G$ in line~\ref{line:pickg}
of Algorithm \emph{split}. This situation can be improved by a
simple caching mechanism: Note that in line~\ref{line:reccall} of Algorithm
\emph{split} we call $\fn{split}(Y,f)$ with affine cells $Y$
satisfying $Y \subset X$. This certainly means
$G = \fn{basis}(X\setminus V(f)) \subseteq \fn{basis}(Y \setminus V(f))$.  Hence we may first
try to pick an element from the already computed set $G$ in 6 of the
call $\fn{split}(Y,f)$ before computing $\fn{basis}(Y\setminus V(f))$.

\subsubsection{Rationale for the new Data Structure}
\label{sec:rat}

Always knowing a Gröbner basis for the affine cells appearing in
Algorithm \emph{equidim} puts a large penalty on the cost of our
algorithms.

This is actually related to well-known observations on the complexity
of Gröbner bases under some regularity assumptions. Indeed, for a
regular sequence in strong Noether position, the cost of linear
algebra steps needed to compute intermediate Gröbner bases in an
incremental manner is higher than the final steps
\parencite{BARDET201549}. Dimension dependent complexity bounds
provide another confirmation of
this behaviour \parencite{10.1145/3087604.3087624}.

Using witness sets we can potentially avoid a lot of intermediate
Gröbner basis computations in our algorithms. In our experience, for a
large number of cases, using witness sets greatly improves the efficiency
of our algorithm which is theoretically backed up by the previously mentioned complexity results.

Furthermore, in the data structure for affine cells presented in the
last subsection we store the definining inequation of our affine cells
as a factorization. If one wants to saturate a polynomial ideal $I$ by
an element $f\in R$ which is known to have a factorization
$f = \prod_{g\in G}g$ given by a finite set $G$ then it is expected to be
cheaper to saturate by the elements $g\in G$ one-by-one instead of
saturating by $f$ directly using the above elimination method. This
lowers the degrees of the polynomials involved.

\subsubsection{A Better Version of \emph{remove}}
\label{sec:rem}

Furthermore, we encountered the following problem when implementing
Algorithm \emph{remove} as presented in Section
\ref{sec:algop}. In this algorithm~$H$ is
a Gröbner basis so it tends to be very redundant (that is very far from being a minimal set of generators).
So it often happens that there are
two or more elements $h_1,h_2\in H$ such that for
$X_1 := X\setminus V(h_1)$ and $X_2:=X\setminus V(h_2)$ we have
$I(X_1) = I(X_2)$. Eventually the sets $X_1$ and $X_2$ become
disjoint, since eventually $X_1$ is intersected with $V(h_1)$ or $X_2$
is intersected with $V(h_1)$, but Algorithm \emph{split} may have to
split $X_1$ and $X_2$ before that happens.  Since splitting an affine
cell with Algorithm \emph{split} depends only on the underlying ideals
we may then repeat the exact same operations on the level of ideals
twice or more. This issue then compounds exponentially due to the
recursive nature of our algorithms. We therefore modified Algorithm
\emph{remove} to obtain disjoint equidimensional affine cells from
$X_1$ and $X_2$ as fast as possible, resulting in
Algorithm~\ref{alg:nznew}. Note that when we use witness sets this
algorithm avoids knowing Gröbner bases for the ideals $I(X_i)$ until
potentially line 12.

\begin{algorithm}[]
  \caption{remove'}
  \label{alg:nznew}
  \raggedright
  \begin{description}[labelsep=1em]
    \item[Input] An affine cell $X$, a finite set $H\subset R$
    \item[Output] A partition of $X\setminus V(H)$ into equidimensional affine cells
  \end{description}
  \begin{pseudo}
    function \fn{remove'}(X, H)\\+
    $\mathcal{D} \gets \emptyset$\\
    for $i$ from $1$ to $r$\\+
    $X_i \gets X\setminus V(h_i)$\\
    $H_i\gets \emptyset$\\
    for $j$ from $1$ to $i-1$\\+
    if \fn{isProper}(X_i, h_j)\\+
     $X_i\gets X_i\cap V(h_j)$\\-
    else\\+
    $H_i\gets H_i\cup \{h_j\}$\\-
    end \\-
    end \\
    $\mathcal{D}_i\gets $ \tn{decomposition of $X_i\cap V(H_i)$}\\ \hspace{2cm}\tn{by repeated application of \emph{split}}\\
    $\mathcal{D} \gets \mathcal{D} \cup \mathcal{D}_i$\\-
    end\\
    return $\mathcal{D}$\\-
    end
  \end{pseudo}
\end{algorithm}

\subsection{Experimental Results}
\label{sec:experi}

In this section we give some experimental results. We compare the
timings of our implementation of Algorithm \emph{equidim} and methods
for equidimensional decomposition of algebraic sets available in
various computer algebra systems in a table given below. Some of the
timings are discussed in more detail in the next section. We compared
to the following implementations:
\begin{enumerate}
\item The function \texttt{Triangularize} from the
  \texttt{RegularChains} library in Maple
  \parencite{lemaire2005},
  which decomposes a polynomial system into regular chains,
\item The function \texttt{equidimensional\_decomposition\_weak} in
  \Oscar{} \parencite{Oscar2023} which is a wrapper around a
  corresponding \texttt{Singular} function \parencite{singular}.
\item The Magma \parencite{bosma1997} functions\\
  \texttt{EquidimensionalDecomposition} (corresponding to the column
  ``Magma'' in the table) and\\
  \texttt{ProbablePrimeDecomposition} (corresponding to the column
  ``Magma (prime dec.)'' in the table) and
\item The numerical polynomial systems solver Bertini \parencite{bates2013}
  which we ran on each system at hand by requesting a witness set
  decomposition into irreducible components with fixed precision set
  to Bertini's default value.
\end{enumerate}

The implementation of our algorithms is itself done in \Oscar{} which
is written in the programming language Julia
\parencite{bezanson2017}. Its source code is available at
\begin{center}
  https://github.com/RafaelDavidMohr/Decomp.jl
\end{center}
For all necessary Gröbner basis computations
we employ the library \msolve{} \parencite{berthomieu2021} for which
\Oscar{} offers an interface.  Our suite of example systems is
comprised as follows:

\begin{enumerate}
\item Cyclic$(8)$, coming from the classical Cyclic$(n)$ benchmark.
\item The systems P4L 1 to 3 come from the perspective-four line problem
  in robotics, see \textcite{garciafontan:hal-03499974}.
\item The systems C1 to C3 are certain jacobian ideals of single
  multivariate polynomials which define singular hypersurfaces.
\item Ps$(n)$, encoding pseudo-singularities via polynomials
  $$f_1,\dots,f_{n-1},g_1,\dots,g_{n-1}$$
  with $f_i\in \mathbb{K}[x_1, \ldots, x_{n-2}, z_1, z_2]$,
  $g_i\in \mathbb{K}[y_1, \ldots, y_{n-2}, z_1, z_2]$, the $f_i$ being chosen
  as a random dense quadrics, and $g_i$ chosen such that
  $g_i(x_1,\dots,x_{n-2},z_1,z_2) = f$, i.e. as a copy of $f_i$ in the
  variables $y_1,\dots,y_{n-2},z_1,z_2$.
\item sos$(s, n)$, encoding the critical points of the restriction of
  the projection on the first coordinate to a hypersurface which is a
  sum of $s$ random dense quadrics in $\mathbb{K}[x_1, \ldots, x_n]$.
  $$f, \frac{\partial f}{\partial x_2}, \ldots, \frac{\partial f}{\partial x_n},
  \quad f = \sum_{i=1}^s g_i^2.$$
\item sing$(n)$, encoding the critical points of the restriction of
  the projection on the first coordinate to a (generically singular)
  hypersurface which is defined by the resultant of two random dense
  quadrics $A, B$ in $\mathbb{K}[x_1, \ldots, x_{n+1}]$:
  $$f, \frac{\partial f}{\partial x_2}, \ldots, \frac{\partial f}{\partial x_n},
  \quad f = \textrm{resultant}(A, B, x_{n+1}).$$
\item The Steiner polynomial system, coming from \textcite{breiding2020}.
\item All remaining examples are part of the BPAS library
  \parencite{asadi2021}. The BPAS library offers an alternative to the
  \texttt{RegularChains} library in Maple with special emphasis on
  paralellism and it will be interesting to compare it to our
  algorithm in the future.
\end{enumerate}

To obtain the timings in the table below we almost exclusively used
the witness set based data structure for affine cells. Every
polynomial system was computed with in characteristic 65521 with the
exception of Bertini which, as a numerical piece of software, computes
over the complex numbers. Due to this difference a comparison between
Bertini's and our timings needs to be considered carefully. We tried
to indicate this in the table below by coloring the Bertini column in
grey. All computations except for Magma were done on an single core of
an Intel Xeon Gold 6244 CPU @ 3.60GHz. All Magma computations were
done on a single core of an Intel Xeon E5-2690 @ 2.90GHz. We let every
algorithm run for at least an hour or 50 times the time it took for
the fastest algorithm to complete the system in question, whichever
was bigger.

Using the witness sets of our output we also did the following to
compare to Bertini: We ran our algorithm in a large random prime
charateristic. We then removed the embedded irreducible components
from each of our output components and computed the degrees of the
output components. This gives us the degree in each dimension of the
algebraic set defined by the input. Whenever Bertini reports different
degrees, we marked it in the respective column. Due to the randomly
chosen large characteristic these degrees should be the same one
obtains when considering the algebraic set in question over the
complex numbers.

In the second column of this table, we additionally provide the number
of affine cells that Algorithm \emph{equidim} decomposed the
respective system into. All timings in this table are given in
seconds. Due to the way we measured the timings of Bertini we can only
report them without any decimal places, rounded up.

\subsection{Discussion of Experimental Results}
\label{sec:disc}

We provide here some further information about some of the examples
and the behaviour of the different implementations on these examples
compared below.

Our algorithm, i.e. Algorithm \emph{equidim}, seems to behave best in
comparison to the other implementations when the input system is dense
in the sense that each of the input equations of the system in
question involves most, or all, of the variables. This is the case for
cyclic 8, the class of the Ps($\bullet$) systems, the class of the
Sing($\bullet$) systems, the class of the sos($\bullet,\bullet$) systems and the Steiner
polynomial system.

On certain polynomial systems, where each input equation involves only
a small subset of the variables, we were able to improve our timings
by foregoing the witness set based data structure and instead running
a deterministic version of our algorithm akin to the version in
Section~\ref{sec:algop}. The improvement we thusly obtained can be
explained by the fact that intersecting very sparse systems with
random hyperplanes can ``destroy their sparsity'' and make certain
Gröbner basis computations much harder. This was the case for the
example Leykin-1: Here running the deterministic version improved our
timing to 2.6 seconds.

The Gonnet and dgp6 polynomial systems demonstrate that our algorithm
is highly sensitive to the ordering of the input equations: By default
we ran our implementation by iterating over the input equations degree
by degree in Algorithm \emph{equidim}. With this ordering, our
algorithm did not terminate within several hours of computation. When
we changed this ordering on these two examples and sorted the input
equations instead by length of support, our algorithm terminated in
less than one second on these two examples. Our algorithms practical
efficiency depends highly on the difficulty of the intermediate
polynomial systems which it encounters, these in turn depend on the
order of the input equations. On the Gonnet polynomial system, the new
ordering resulted in only a subset of the variables being involved in
the first few intermediate systems, thus making Gröbner basis for them
more tractable. On the dgp6 example the new ordering resulted in an
intermediate polynomial system consisting of monomials and binomials
which our algorithm decomposes very finely, making the treatment of
the remaining equations substantially easier

The system sys2874 can be attacked by both changing the order of the
input equations to be ordered by length of support and by using the
deterministic version of our algorithm: Doing this, the timing
improved by several orders of magnitude to 0.26 seconds.

We also remark that \Oscar's timings improved significantly on
the examples sys2449, sys2297 and Leykin-1 (each to less than one
second) if one decomposes the radicals of these systems instead of the
systems themselves.

For the examples KdV and sys2882 we seem to be bottlenecked by very
difficult Gröbner basis computations and less by the inherent
structure of our algorithm. Informal experiments where we tried to
compute just a Gröbner basis for these systems using msolve suggest
that even this is a highly non-trivial computation. For these two
systems, techniques involving regular chains seem to be vastly
superior over anything that involves Gröbner basis computations.

All in all, these experiments illustrate that on a wide range of
examples, our algorithm performs on average better than
state-of-the-art implementations and can tackle some problems which
were previously unrea
\section*{Acknowledgements}
  The authors wish to thank Marc Moreno Maza for providing feedback on
  the BPAS library and its benchmark examples. Additionally the
  authors thank the two anonymous reviewers for their helpful comments
  and careful reading of this manuscript.
  
  This work has been supported by the Agence nationale de la recherche
  (ANR), grant agreements ANR-18-CE33-0011 (SESAME), ANR-19-CE40-0018
  (Re Rerum Natura); by the joint ANR-Austrian Science Fund FWF grant
  agreements ANR-19-CE48-0015 (ECARP) and ANR-22-CE91-0007 (EAGLES);
  by the EOARD-AFOSR grant agreement FA8665-20-1-7029; by the DFG
  Sonderforschungsbereich TRR 195; the Forschungsinitiative
  Rheinland-Pfalz; and by the European Research Council (ERC) under
  the European Union’s Horizon Europe research and innovation
  programme, grant agreement 101040794 (10000~DIGITS).

\newpage
% Compile the table with
% python table_formatting.py < table.csv > benchmark.tex
\begin{table*}[p]
%  \caption{Comparing with other Decomposition Methods}
  \label{tbl:compare}
  \begin{adjustwidth}{-1.5cm}{-1cm}
    \begin{center}
      \scriptsize{\begin{tabular}{llllllll}
\toprule name & nb. comp. & \fn{equidim} & Maple & Oscar & Magma & Magma (prime dec.) & \color{Gray} Bertini\\ \midrule
8-3-config-Li & 23 & \bfseries\color{SeaGreen} {\small \phantom{\faWarning}}\phantom{000}1.6 & \bfseries\color{BurntOrange} {\small \phantom{\faWarning}}\phantom{00}16\phantom{.0} \footnotesize{\mdseries ×\phantom{}10\phantom{.0}} & \bfseries\color{BrickRed} {\small \phantom{\faWarning}}\phantom{0000}\llap{>\,1h} & \bfseries\color{BrickRed} {\small \phantom{\faWarning}}\phantom{0000}\llap{>\,1h} & \bfseries\color{BrickRed} {\small \phantom{\faWarning}}\phantom{00}65\phantom{.0} \footnotesize{\mdseries ×\phantom{}40\phantom{.0}} & \bfseries\color{Gray} {\small \phantom{\faWarning}}\phantom{000}4\phantom{.0} \footnotesize{\mdseries ×\phantom{}2.5}\\
cyclic8 & 6 & \bfseries\color{SeaGreen} {\small \phantom{\faWarning}}\phantom{0}381\phantom{.0} & \bfseries\color{BrickRed} {\small \phantom{\faWarning}}\phantom{0000}\llap{>\,5h} & \bfseries\color{BrickRed} {\small \phantom{\faWarning}}\phantom{0000}\llap{>\,5h} & \bfseries\color{BrickRed} {\small \phantom{\faWarning}}\phantom{0000}\llap{>\,5h} & \bfseries\color{BrickRed} {\small \phantom{\faWarning}}\phantom{0000}\llap{>\,5h} & \bfseries\color{Gray} {\small \faWarning}\phantom{0}126\phantom{.0} \footnotesize{\mdseries ×\phantom{}0.3}\\
dgp6 & 3 & \bfseries\color{SeaGreen} {\small \faHandPaperO}\phantom{000}0.2 & \bfseries\color{BrickRed} {\small \phantom{\faWarning}}\phantom{00}53\phantom{.0} & \bfseries\color{BurntOrange} {\small \phantom{\faWarning}}\phantom{000}2.2 & \bfseries\color{BrickRed} {\small \phantom{\faWarning}}\phantom{0000}\llap{>\,1h} & \bfseries\color{OliveGreen} {\small \phantom{\faWarning}}\phantom{000}1.2 & \bfseries\color{Gray} {\small \phantom{\faWarning}}\phantom{00}75\phantom{.0}\\
Gonnet & 3 & \bfseries\color{SeaGreen} {\small \faHandPaperO}\phantom{000}0.2 & \bfseries\color{BurntOrange} {\small \phantom{\faWarning}}\phantom{000}2.1 & \bfseries\color{BurntOrange} {\small \phantom{\faWarning}}\phantom{000}2.8 & \bfseries\color{BrickRed} {\small \phantom{\faWarning}}\phantom{0000}\llap{>\,1h} & \bfseries\color{OliveGreen} {\small \phantom{\faWarning}}\phantom{000}1.4 & \bfseries\color{Gray} {\small \phantom{\faWarning}}\phantom{00}74\phantom{.0}\\
P4L1 & 6 & \bfseries\color{SeaGreen} {\small \phantom{\faWarning}}\phantom{000}0.3 & \bfseries\color{BurntOrange} {\small \phantom{\faWarning}}\phantom{000}2.4 & \bfseries\color{OliveGreen} {\small \phantom{\faWarning}}\phantom{000}1.8 & \bfseries\color{OliveGreen} {\small \phantom{\faWarning}}\phantom{000}0.7 & \bfseries\color{OliveGreen} {\small \phantom{\faWarning}}\phantom{000}1.5 & \bfseries\color{Gray} {\small \faWarning}\phantom{00}21\phantom{.0}\\
P4L3 & 8 & \bfseries\color{OliveGreen} {\small \phantom{\faWarning}}\phantom{000}0.3 & \bfseries\color{BurntOrange} {\small \phantom{\faWarning}}\phantom{000}3.3 & \bfseries\color{BurntOrange} {\small \phantom{\faWarning}}\phantom{000}10\phantom{.0} & \bfseries\color{SeaGreen} {\small \phantom{\faWarning}}\phantom{000}\llap{<\,}0.1 & \bfseries\color{OliveGreen} {\small \phantom{\faWarning}}\phantom{000}1.5 & \bfseries\color{Gray} {\small \phantom{\faWarning}}\phantom{00}11\phantom{.0}\\
KdV &  & \bfseries\color{BrickRed} {\small \phantom{\faWarning}}\phantom{0000}\llap{>\,4h} & \bfseries\color{SeaGreen} {\small \phantom{\faWarning}}\phantom{0}353\phantom{.0} & \bfseries\color{BrickRed} {\small \phantom{\faWarning}}\phantom{0000}\llap{>\,4h} & \bfseries\color{BrickRed} {\small \phantom{\faWarning}}\phantom{0000}\llap{>\,4h} & \bfseries\color{BrickRed} {\small \phantom{\faWarning}}\phantom{}7109\phantom{.0} \footnotesize{\mdseries ×\phantom{}20\phantom{.0}} & \bfseries\color{Gray} {\small \phantom{\faWarning}}\phantom{0000}\llap{>\,4h}\\
Leykin-1 & 13 & \bfseries\color{OliveGreen} {\small \faHandPaperO}\phantom{000}2.6 \footnotesize{\mdseries ×\phantom{}1.9} & \bfseries\color{BurntOrange} {\small \phantom{\faWarning}}\phantom{000}4\phantom{.0} \footnotesize{\mdseries ×\phantom{}3.2} & \bfseries\color{BrickRed} {\small \phantom{\faWarning}}\phantom{0}641\phantom{.0} \footnotesize{\mdseries ×\phantom{}468\phantom{.0}} & \bfseries\color{BrickRed} {\small \phantom{\faWarning}}\phantom{0000}\llap{>\,1h} & \bfseries\color{SeaGreen} {\small \phantom{\faWarning}}\phantom{000}1.4 & \color{Gray} \small \faClose\\
C1 & 4 & \bfseries\color{SeaGreen} {\small \phantom{\faWarning}}\phantom{0}129\phantom{.0} & \bfseries\color{BrickRed} {\small \phantom{\faWarning}}\phantom{0000}\llap{>\,1h} & \bfseries\color{BrickRed} {\small \phantom{\faWarning}}\phantom{0000}\llap{>\,1h} & \bfseries\color{BrickRed} {\small \phantom{\faWarning}}\phantom{0000}\llap{>\,1h} & \bfseries\color{BrickRed} {\small \phantom{\faWarning}}\phantom{0000}\llap{>\,1h} & \color{Gray} \small \faClose\\
C2 & 4 & \bfseries\color{SeaGreen} {\small \phantom{\faWarning}}\phantom{000}0.3 & \bfseries\color{BrickRed} {\small \phantom{\faWarning}}\phantom{0}100\phantom{.0} & \bfseries\color{BrickRed} {\small \phantom{\faWarning}}\phantom{0}152\phantom{.0} & \bfseries\color{BrickRed} {\small \phantom{\faWarning}}\phantom{0000}\llap{>\,1h} & \bfseries\color{OliveGreen} {\small \phantom{\faWarning}}\phantom{000}2.0 & \color{Gray} \small \faClose\\
C3 & 13 & \bfseries\color{BurntOrange} {\small \phantom{\faWarning}}\phantom{000}10\phantom{.0} & \bfseries\color{BrickRed} {\small \phantom{\faWarning}}\phantom{00}55\phantom{.0} & \bfseries\color{BurntOrange} {\small \phantom{\faWarning}}\phantom{000}7\phantom{.0} & \bfseries\color{SeaGreen} {\small \phantom{\faWarning}}\phantom{000}0.3 & \bfseries\color{OliveGreen} {\small \phantom{\faWarning}}\phantom{000}1.5 & \color{Gray} \small \faClose\\
MontesS16 & 6 & \bfseries\color{OliveGreen} {\small \phantom{\faWarning}}\phantom{000}1.9 \footnotesize{\mdseries ×\phantom{}1.4} & \bfseries\color{OliveGreen} {\small \phantom{\faWarning}}\phantom{000}2.7 \footnotesize{\mdseries ×\phantom{}1.9} & \bfseries\color{OliveGreen} {\small \phantom{\faWarning}}\phantom{000}2.0 \footnotesize{\mdseries ×\phantom{}1.4} & \bfseries\color{SeaGreen} {\small \phantom{\faWarning}}\phantom{000}1.4 & \bfseries\color{OliveGreen} {\small \phantom{\faWarning}}\phantom{000}1.5 \footnotesize{\mdseries ×\phantom{}1.1} & \bfseries\color{Gray} {\small \phantom{\faWarning}}\phantom{000}7\phantom{.0} \footnotesize{\mdseries ×\phantom{}5\phantom{.0}}\\
Ps(10) & 2 & \bfseries\color{SeaGreen} {\small \phantom{\faWarning}}\phantom{000}1.7 & \bfseries\color{BrickRed} {\small \phantom{\faWarning}}\phantom{0000}\llap{>\,1h} & \bfseries\color{BrickRed} {\small \phantom{\faWarning}}\phantom{00}30\phantom{.0} \footnotesize{\mdseries ×\phantom{}17\phantom{.0}} & \bfseries\color{BrickRed} {\small \phantom{\faWarning}}\phantom{0000}\llap{>\,1h} & \bfseries\color{BurntOrange} {\small \phantom{\faWarning}}\phantom{000}6\phantom{.0} \footnotesize{\mdseries ×\phantom{}3.3} & \bfseries\color{Gray} {\small \phantom{\faWarning}}\phantom{000}9\phantom{.0} \footnotesize{\mdseries ×\phantom{}5\phantom{.0}}\\
Ps(12) & 2 & \bfseries\color{SeaGreen} {\small \phantom{\faWarning}}\phantom{00}51\phantom{.0} & \bfseries\color{BrickRed} {\small \phantom{\faWarning}}\phantom{0000}\llap{>\,1h} & \bfseries\color{BrickRed} {\small \phantom{\faWarning}}\phantom{0000}\llap{>\,1h} & \bfseries\color{BrickRed} {\small \phantom{\faWarning}}\phantom{0000}\llap{>\,1h} & \bfseries\color{BrickRed} {\small \phantom{\faWarning}}\phantom{}2060\phantom{.0} \footnotesize{\mdseries ×\phantom{}40\phantom{.0}} & \bfseries\color{Gray} {\small \faWarning}\phantom{00}38\phantom{.0} \footnotesize{\mdseries ×\phantom{}0.7}\\
Ps(6) & 2 & \bfseries\color{SeaGreen} {\small \phantom{\faWarning}}\phantom{000}\llap{<\,}0.1 & \bfseries\color{OliveGreen} {\small \phantom{\faWarning}}\phantom{000}0.2 & \bfseries\color{OliveGreen} {\small \phantom{\faWarning}}\phantom{000}1.7 & \bfseries\color{OliveGreen} {\small \phantom{\faWarning}}\phantom{000}0.5 & \bfseries\color{OliveGreen} {\small \phantom{\faWarning}}\phantom{000}0.3 & \bfseries\color{Gray} {\small \phantom{\faWarning}}\phantom{000}2.0\\
Ps(8) & 2 & \bfseries\color{SeaGreen} {\small \phantom{\faWarning}}\phantom{000}\llap{<\,}0.1 & \bfseries\color{BurntOrange} {\small \phantom{\faWarning}}\phantom{000}4\phantom{.0} & \bfseries\color{OliveGreen} {\small \phantom{\faWarning}}\phantom{000}1.7 & \bfseries\color{OliveGreen} {\small \phantom{\faWarning}}\phantom{000}1.2 & \bfseries\color{OliveGreen} {\small \phantom{\faWarning}}\phantom{000}0.8 & \bfseries\color{Gray} {\small \phantom{\faWarning}}\phantom{000}6\phantom{.0}\\
Sing(10) & 2 & \bfseries\color{SeaGreen} {\small \phantom{\faWarning}}\phantom{000}0.4 & \bfseries\color{BrickRed} {\small \phantom{\faWarning}}\phantom{0000}\llap{>\,1h} & \bfseries\color{BrickRed} {\small \phantom{\faWarning}}\phantom{0000}\llap{>\,1h} & \bfseries\color{BrickRed} {\small \phantom{\faWarning}}\phantom{0000}\llap{>\,1h} & \bfseries\color{BrickRed} {\small \phantom{\faWarning}}\phantom{0000}\llap{>\,1h} & \bfseries\color{Gray} {\small \faWarning}\phantom{0}495\phantom{.0}\\
Sing(4) & 2 & \bfseries\color{SeaGreen} {\small \phantom{\faWarning}}\phantom{000}\llap{<\,}0.1 & \bfseries\color{BrickRed} {\small \phantom{\faWarning}}\phantom{00}76\phantom{.0} & \bfseries\color{BurntOrange} {\small \phantom{\faWarning}}\phantom{000}2.2 & \bfseries\color{BurntOrange} {\small \phantom{\faWarning}}\phantom{000}8\phantom{.0} & \bfseries\color{BurntOrange} {\small \phantom{\faWarning}}\phantom{000}5\phantom{.0} & \bfseries\color{Gray} {\small \phantom{\faWarning}}\phantom{000}5\phantom{.0}\\
Sing(5) & 2 & \bfseries\color{SeaGreen} {\small \phantom{\faWarning}}\phantom{000}\llap{<\,}0.1 & \bfseries\color{BrickRed} {\small \phantom{\faWarning}}\phantom{0000}\llap{>\,1h} & \bfseries\color{BurntOrange} {\small \phantom{\faWarning}}\phantom{000}4\phantom{.0} & \bfseries\color{BurntOrange} {\small \phantom{\faWarning}}\phantom{000}7\phantom{.0} & \bfseries\color{BrickRed} {\small \phantom{\faWarning}}\phantom{}1636\phantom{.0} & \bfseries\color{Gray} {\small \faWarning}\phantom{000}1.0\\
Sing(6) & 2 & \bfseries\color{SeaGreen} {\small \phantom{\faWarning}}\phantom{000}\llap{<\,}0.1 & \bfseries\color{BrickRed} {\small \phantom{\faWarning}}\phantom{0000}\llap{>\,1h} & \bfseries\color{BrickRed} {\small \phantom{\faWarning}}\phantom{00}51\phantom{.0} & \bfseries\color{BrickRed} {\small \phantom{\faWarning}}\phantom{0000}\llap{>\,1h} & \bfseries\color{BrickRed} {\small \phantom{\faWarning}}\phantom{0000}\llap{>\,1h} & \bfseries\color{Gray} {\small \faWarning}\phantom{000}8\phantom{.0}\\
Sing(7) & 2 & \bfseries\color{SeaGreen} {\small \phantom{\faWarning}}\phantom{000}\llap{<\,}0.1 & \bfseries\color{BrickRed} {\small \phantom{\faWarning}}\phantom{}1704\phantom{.0} & \bfseries\color{BrickRed} {\small \phantom{\faWarning}}\phantom{0}399\phantom{.0} & \bfseries\color{BrickRed} {\small \phantom{\faWarning}}\phantom{0000}\llap{>\,1h} & \bfseries\color{BrickRed} {\small \phantom{\faWarning}}\phantom{0000}\llap{>\,1h} & \bfseries\color{Gray} {\small \phantom{\faWarning}}\phantom{00}54\phantom{.0}\\
Sing(8) & 2 & \bfseries\color{SeaGreen} {\small \phantom{\faWarning}}\phantom{000}0.1 & \bfseries\color{BrickRed} {\small \phantom{\faWarning}}\phantom{0000}\llap{>\,1h} & \bfseries\color{BrickRed} {\small \phantom{\faWarning}}\phantom{0}995\phantom{.0} & \bfseries\color{BrickRed} {\small \phantom{\faWarning}}\phantom{0000}\llap{>\,1h} & \bfseries\color{BrickRed} {\small \phantom{\faWarning}}\phantom{0000}\llap{>\,1h} & \bfseries\color{Gray} {\small \phantom{\faWarning}}\phantom{0}139\phantom{.0}\\
Sing(9) & 2 & \bfseries\color{SeaGreen} {\small \phantom{\faWarning}}\phantom{000}0.2 & \bfseries\color{BrickRed} {\small \phantom{\faWarning}}\phantom{0000}\llap{>\,1h} & \bfseries\color{BrickRed} {\small \phantom{\faWarning}}\phantom{0000}\llap{>\,1h} & \bfseries\color{BrickRed} {\small \phantom{\faWarning}}\phantom{0000}\llap{>\,1h} & \bfseries\color{BrickRed} {\small \phantom{\faWarning}}\phantom{0000}\llap{>\,1h} & \bfseries\color{Gray} {\small \faWarning}\phantom{0}271\phantom{.0}\\
sos(4,2) & 2 & \bfseries\color{SeaGreen} {\small \phantom{\faWarning}}\phantom{000}\llap{<\,}0.1 & \bfseries\color{BrickRed} {\small \phantom{\faWarning}}\phantom{00}16\phantom{.0} & \bfseries\color{BurntOrange} {\small \phantom{\faWarning}}\phantom{000}2.2 & \bfseries\color{OliveGreen} {\small \phantom{\faWarning}}\phantom{000}1.3 & \bfseries\color{OliveGreen} {\small \phantom{\faWarning}}\phantom{000}1.2 & \bfseries\color{Gray} {\small \phantom{\faWarning}}\phantom{000}1.0\\
sos(4,3) & 2 & \bfseries\color{SeaGreen} {\small \phantom{\faWarning}}\phantom{000}\llap{<\,}0.1 & \bfseries\color{BrickRed} {\small \phantom{\faWarning}}\phantom{0}694\phantom{.0} & \bfseries\color{BurntOrange} {\small \phantom{\faWarning}}\phantom{000}2.6 & \bfseries\color{BurntOrange} {\small \phantom{\faWarning}}\phantom{000}3.4 & \bfseries\color{BurntOrange} {\small \phantom{\faWarning}}\phantom{000}6\phantom{.0} & \bfseries\color{Gray} {\small \phantom{\faWarning}}\phantom{000}3.0\\
sos(5,2) & 2 & \bfseries\color{SeaGreen} {\small \phantom{\faWarning}}\phantom{000}\llap{<\,}0.1 & \bfseries\color{BrickRed} {\small \phantom{\faWarning}}\phantom{0000}\llap{>\,1h} & \bfseries\color{OliveGreen} {\small \phantom{\faWarning}}\phantom{000}1.8 & \bfseries\color{BurntOrange} {\small \phantom{\faWarning}}\phantom{000}3.7 & \bfseries\color{OliveGreen} {\small \phantom{\faWarning}}\phantom{000}1.2 & \bfseries\color{Gray} {\small \phantom{\faWarning}}\phantom{000}3.0\\
sos(5,3) & 2 & \bfseries\color{SeaGreen} {\small \phantom{\faWarning}}\phantom{000}\llap{<\,}0.1 & \bfseries\color{BrickRed} {\small \phantom{\faWarning}}\phantom{0000}\llap{>\,1h} & \bfseries\color{BrickRed} {\small \phantom{\faWarning}}\phantom{0000}\llap{>\,1h} & \bfseries\color{BrickRed} {\small \phantom{\faWarning}}\phantom{0000}\llap{>\,1h} & \bfseries\color{BrickRed} {\small \phantom{\faWarning}}\phantom{0}149\phantom{.0} & \bfseries\color{Gray} {\small \phantom{\faWarning}}\phantom{00}15\phantom{.0}\\
sos(5,4) & 2 & \bfseries\color{SeaGreen} {\small \phantom{\faWarning}}\phantom{000}0.5 & \bfseries\color{BrickRed} {\small \phantom{\faWarning}}\phantom{0000}\llap{>\,1h} & \bfseries\color{BrickRed} {\small \phantom{\faWarning}}\phantom{0000}\llap{>\,1h} & \bfseries\color{BrickRed} {\small \phantom{\faWarning}}\phantom{0000}\llap{>\,1h} & \bfseries\color{BrickRed} {\small \phantom{\faWarning}}\phantom{0000}\llap{>\,1h} & \bfseries\color{Gray} {\small \phantom{\faWarning}}\phantom{00}21\phantom{.0}\\
sos(6,2) & 2 & \bfseries\color{SeaGreen} {\small \phantom{\faWarning}}\phantom{000}\llap{<\,}0.1 & \bfseries\color{BrickRed} {\small \phantom{\faWarning}}\phantom{0000}\llap{>\,1h} & \bfseries\color{BurntOrange} {\small \phantom{\faWarning}}\phantom{000}2.0 & \bfseries\color{BurntOrange} {\small \phantom{\faWarning}}\phantom{000}5\phantom{.0} & \bfseries\color{OliveGreen} {\small \phantom{\faWarning}}\phantom{000}1.6 & \bfseries\color{Gray} {\small \phantom{\faWarning}}\phantom{000}5\phantom{.0}\\
sos(6,3) & 2 & \bfseries\color{SeaGreen} {\small \phantom{\faWarning}}\phantom{000}0.1 & \bfseries\color{BrickRed} {\small \phantom{\faWarning}}\phantom{0000}\llap{>\,1h} & \bfseries\color{BrickRed} {\small \phantom{\faWarning}}\phantom{0000}\llap{>\,1h} & \bfseries\color{BrickRed} {\small \phantom{\faWarning}}\phantom{0000}\llap{>\,1h} & \bfseries\color{BrickRed} {\small \phantom{\faWarning}}\phantom{0000}\llap{>\,1h} & \bfseries\color{Gray} {\small \phantom{\faWarning}}\phantom{00}34\phantom{.0}\\
sos(6,4) & 2 & \bfseries\color{SeaGreen} {\small \phantom{\faWarning}}\phantom{000}5\phantom{.0} & \bfseries\color{BrickRed} {\small \phantom{\faWarning}}\phantom{0000}\llap{>\,1h} & \bfseries\color{BrickRed} {\small \phantom{\faWarning}}\phantom{0000}\llap{>\,1h} & \bfseries\color{BrickRed} {\small \phantom{\faWarning}}\phantom{0000}\llap{>\,1h} & \bfseries\color{BrickRed} {\small \phantom{\faWarning}}\phantom{0000}\llap{>\,1h} & \bfseries\color{Gray} {\small \phantom{\faWarning}}\phantom{00}69\phantom{.0} \footnotesize{\mdseries ×\phantom{}14\phantom{.0}}\\
sos(6,5) & 2 & \bfseries\color{SeaGreen} {\small \phantom{\faWarning}}\phantom{00}14\phantom{.0} & \bfseries\color{BrickRed} {\small \phantom{\faWarning}}\phantom{0000}\llap{>\,1h} & \bfseries\color{BrickRed} {\small \phantom{\faWarning}}\phantom{0000}\llap{>\,1h} & \bfseries\color{BrickRed} {\small \phantom{\faWarning}}\phantom{0000}\llap{>\,1h} & \bfseries\color{BrickRed} {\small \phantom{\faWarning}}\phantom{0000}\llap{>\,1h} & \bfseries\color{Gray} {\small \faWarning}\phantom{00}40\phantom{.0} \footnotesize{\mdseries ×\phantom{}2.9}\\
steiner & 2 & \bfseries\color{SeaGreen} {\small \phantom{\faWarning}}\phantom{0}870\phantom{.0} & \bfseries\color{BrickRed} {\small \phantom{\faWarning}}\phantom{0000}\llap{>\,12h} & \bfseries\color{BrickRed} {\small \phantom{\faWarning}}\phantom{0000}\llap{>\,12h} & \bfseries\color{BrickRed} {\small \phantom{\faWarning}}\phantom{0000}\llap{>\,12h} & \bfseries\color{BrickRed} {\small \phantom{\faWarning}}\phantom{0000}\llap{>\,12h} & \color{Gray} \small \faClose\\
sys2128 & 20 & \bfseries\color{SeaGreen} {\small \phantom{\faWarning}}\phantom{000}1.1 & \bfseries\color{BurntOrange} {\small \phantom{\faWarning}}\phantom{000}9\phantom{.0} \footnotesize{\mdseries ×\phantom{}9\phantom{.0}} & \bfseries\color{BurntOrange} {\small \phantom{\faWarning}}\phantom{000}5\phantom{.0} \footnotesize{\mdseries ×\phantom{}4\phantom{.0}} & \bfseries\color{BrickRed} {\small \phantom{\faWarning}}\phantom{0000}\llap{>\,1h} & \bfseries\color{OliveGreen} {\small \phantom{\faWarning}}\phantom{000}1.9 \footnotesize{\mdseries ×\phantom{}1.8} & \bfseries\color{Gray} {\small \phantom{\faWarning}}\phantom{0}829\phantom{.0} \footnotesize{\mdseries ×\phantom{}790\phantom{.0}}\\
sys2161 & 33 & \bfseries\color{SeaGreen} {\small \phantom{\faWarning}}\phantom{000}8\phantom{.0} & \bfseries\color{BurntOrange} {\small \phantom{\faWarning}}\phantom{00}29\phantom{.0} \footnotesize{\mdseries ×\phantom{}3.8} & \bfseries\color{BrickRed} {\small \phantom{\faWarning}}\phantom{0000}\llap{>\,1h} & \bfseries\color{BrickRed} {\small \phantom{\faWarning}}\phantom{0000}\llap{>\,1h} & \bfseries\color{OliveGreen} {\small \phantom{\faWarning}}\phantom{000}8\phantom{.0} \footnotesize{\mdseries ×\phantom{}1.0} & \bfseries\color{Gray} {\small \phantom{\faWarning}}\phantom{}1196\phantom{.0} \footnotesize{\mdseries ×\phantom{}159\phantom{.0}}\\
sys2297 & 11 & \bfseries\color{SeaGreen} {\small \phantom{\faWarning}}\phantom{000}0.5 & \bfseries\color{BrickRed} {\small \phantom{\faWarning}}\phantom{00}14\phantom{.0} & \bfseries\color{BrickRed} {\small \phantom{\faWarning}}\phantom{00}49\phantom{.0} & \bfseries\color{BrickRed} {\small \phantom{\faWarning}}\phantom{0000}\llap{>\,1h} & \bfseries\color{OliveGreen} {\small \phantom{\faWarning}}\phantom{000}1.7 & \bfseries\color{Gray} {\small \faWarning}\phantom{0}497\phantom{.0}\\
sys2353 & 13 & \bfseries\color{OliveGreen} {\small \phantom{\faWarning}}\phantom{000}1.6 \footnotesize{\mdseries ×\phantom{}1.2} & \bfseries\color{BurntOrange} {\small \phantom{\faWarning}}\phantom{000}5\phantom{.0} \footnotesize{\mdseries ×\phantom{}3.8} & \bfseries\color{OliveGreen} {\small \phantom{\faWarning}}\phantom{000}2.0 \footnotesize{\mdseries ×\phantom{}1.5} & \bfseries\color{BurntOrange} {\small \phantom{\faWarning}}\phantom{000}5\phantom{.0} \footnotesize{\mdseries ×\phantom{}3.7} & \bfseries\color{SeaGreen} {\small \phantom{\faWarning}}\phantom{000}1.3 & \bfseries\color{Gray} {\small \phantom{\faWarning}}\phantom{0000}\llap{>\,1h}\\
sys2449 & 24 & \bfseries\color{SeaGreen} {\small \phantom{\faWarning}}\phantom{000}1.3 & \bfseries\color{BrickRed} {\small \phantom{\faWarning}}\phantom{00}28\phantom{.0} \footnotesize{\mdseries ×\phantom{}21\phantom{.0}} & \bfseries\color{BrickRed} {\small \phantom{\faWarning}}\phantom{00}60\phantom{.0} \footnotesize{\mdseries ×\phantom{}46\phantom{.0}} & \bfseries\color{BrickRed} {\small \phantom{\faWarning}}\phantom{0000}\llap{>\,1h} & \bfseries\color{OliveGreen} {\small \phantom{\faWarning}}\phantom{000}2.2 \footnotesize{\mdseries ×\phantom{}1.7} & \bfseries\color{Gray} {\small \phantom{\faWarning}}\phantom{}1338\phantom{.0} \footnotesize{\mdseries ×\phantom{}1014\phantom{.0}}\\
sys2647 & 2 & \bfseries\color{SeaGreen} {\small \phantom{\faWarning}}\phantom{000}\llap{<\,}0.1 & \bfseries\color{BurntOrange} {\small \phantom{\faWarning}}\phantom{000}7\phantom{.0} & \bfseries\color{BurntOrange} {\small \phantom{\faWarning}}\phantom{000}4\phantom{.0} & \bfseries\color{BrickRed} {\small \phantom{\faWarning}}\phantom{0000}\llap{>\,1h} & \bfseries\color{OliveGreen} {\small \phantom{\faWarning}}\phantom{000}2.0 & \bfseries\color{Gray} {\small \faWarning}\phantom{000}9\phantom{.0}\\
sys2874 & 5 & \bfseries\color{SeaGreen} {\small \faHandPaperO}\phantom{000}0.3 & \bfseries\color{BrickRed} {\small \phantom{\faWarning}}\phantom{0}202\phantom{.0} & \bfseries\color{OliveGreen} {\small \phantom{\faWarning}}\phantom{000}1.9 & \bfseries\color{BurntOrange} {\small \phantom{\faWarning}}\phantom{000}8\phantom{.0} & \bfseries\color{BurntOrange} {\small \phantom{\faWarning}}\phantom{000}10\phantom{.0} & \bfseries\color{Gray} {\small \phantom{\faWarning}}\phantom{0000}\llap{>\,1h}\\
sys2880 & 50 & \bfseries\color{BurntOrange} {\small \phantom{\faWarning}}\phantom{000}4\phantom{.0} \footnotesize{\mdseries ×\phantom{}2.4} & \bfseries\color{BrickRed} {\small \phantom{\faWarning}}\phantom{0}144\phantom{.0} \footnotesize{\mdseries ×\phantom{}80\phantom{.0}} & \bfseries\color{SeaGreen} {\small \phantom{\faWarning}}\phantom{000}1.8 & \bfseries\color{OliveGreen} {\small \phantom{\faWarning}}\phantom{000}3.4 \footnotesize{\mdseries ×\phantom{}1.9} & \bfseries\color{BurntOrange} {\small \phantom{\faWarning}}\phantom{000}4\phantom{.0} \footnotesize{\mdseries ×\phantom{}2.2} & \bfseries\color{Gray} {\small \faWarning}\phantom{0}324\phantom{.0} \footnotesize{\mdseries ×\phantom{}180\phantom{.0}}\\
sys2882 &  & \bfseries\color{BrickRed} {\small \phantom{\faWarning}}\phantom{0000}\llap{>\,1h} & \bfseries\color{SeaGreen} {\small \phantom{\faWarning}}\phantom{00}39\phantom{.0} & \bfseries\color{BrickRed} {\small \phantom{\faWarning}}\phantom{0000}\llap{>\,1h} & \bfseries\color{BrickRed} {\small \phantom{\faWarning}}\phantom{0000}\llap{>\,1h} & \bfseries\color{BrickRed} {\small \phantom{\faWarning}}\phantom{0000}\llap{>\,1h} & \color{Gray} \small \faClose\\
sys2885 & 2 & \bfseries\color{SeaGreen} {\small \phantom{\faWarning}}\phantom{000}0.3 & \bfseries\color{BurntOrange} {\small \phantom{\faWarning}}\phantom{000}6\phantom{.0} & \bfseries\color{BurntOrange} {\small \phantom{\faWarning}}\phantom{000}2.2 & \bfseries\color{BrickRed} {\small \phantom{\faWarning}}\phantom{0000}\llap{>\,1h} & \bfseries\color{OliveGreen} {\small \phantom{\faWarning}}\phantom{000}1.2 & \bfseries\color{Gray} {\small \phantom{\faWarning}}\phantom{00}76\phantom{.0}\\
sys2945 & 5 & \bfseries\color{SeaGreen} {\small \phantom{\faWarning}}\phantom{000}0.5 & \bfseries\color{BurntOrange} {\small \phantom{\faWarning}}\phantom{000}3.2 & \bfseries\color{OliveGreen} {\small \phantom{\faWarning}}\phantom{000}1.8 & \bfseries\color{OliveGreen} {\small \phantom{\faWarning}}\phantom{000}0.6 & \bfseries\color{OliveGreen} {\small \phantom{\faWarning}}\phantom{000}1.0 & \bfseries\color{Gray} {\small \phantom{\faWarning}}\phantom{0}120\phantom{.0}\\
sys2946 & 7 & \bfseries\color{SeaGreen} {\small \phantom{\faWarning}}\phantom{000}0.2 & \bfseries\color{OliveGreen} {\small \phantom{\faWarning}}\phantom{000}0.7 & \bfseries\color{BurntOrange} {\small \phantom{\faWarning}}\phantom{000}2.1 & \bfseries\color{BrickRed} {\small \phantom{\faWarning}}\phantom{0000}\llap{>\,1h} & \bfseries\color{OliveGreen} {\small \phantom{\faWarning}}\phantom{000}1.6 & \bfseries\color{Gray} {\small \phantom{\faWarning}}\phantom{000}3.0\\
W2 & 4 & \bfseries\color{SeaGreen} {\small \phantom{\faWarning}}\phantom{000}0.9 & \bfseries\color{BurntOrange} {\small \phantom{\faWarning}}\phantom{000}6\phantom{.0} & \bfseries\color{OliveGreen} {\small \phantom{\faWarning}}\phantom{000}1.9 & \bfseries\color{BurntOrange} {\small \phantom{\faWarning}}\phantom{000}7\phantom{.0} & \bfseries\color{OliveGreen} {\small \phantom{\faWarning}}\phantom{000}1.0 & \bfseries\color{Gray} {\small \phantom{\faWarning}}\phantom{00}61\phantom{.0}\\
W44 & 3 & \bfseries\color{SeaGreen} {\small \phantom{\faWarning}}\phantom{000}0.5 & \bfseries\color{BrickRed} {\small \phantom{\faWarning}}\phantom{00}13\phantom{.0} & \bfseries\color{BurntOrange} {\small \phantom{\faWarning}}\phantom{000}4\phantom{.0} & \bfseries\color{BrickRed} {\small \phantom{\faWarning}}\phantom{0000}\llap{>\,1h} & \bfseries\color{OliveGreen} {\small \phantom{\faWarning}}\phantom{000}1.5 & \bfseries\color{Gray} {\small \faWarning}\phantom{00}66\phantom{.0}\\
Wu-Wang & 3 & \bfseries\color{BurntOrange} {\small \phantom{\faWarning}}\phantom{000}3.1 & \bfseries\color{BurntOrange} {\small \phantom{\faWarning}}\phantom{000}3.2 & \bfseries\color{OliveGreen} {\small \phantom{\faWarning}}\phantom{000}1.8 & \bfseries\color{SeaGreen} {\small \phantom{\faWarning}}\phantom{000}0.7 & \bfseries\color{OliveGreen} {\small \phantom{\faWarning}}\phantom{000}1.3 & \bfseries\color{Gray} {\small \phantom{\faWarning}}\phantom{0}106\phantom{.0}\\
\bottomrule\end{tabular}
}
    \end{center}
  \end{adjustwidth}

    \bigskip
    Timings are in seconds, except otherwise indicated. The ratio with respect to the best time is given when the latter is over 1~second.

    \begin{description}[labelsep=1em]
      \item[\faHandPaperO] We made some minor preparation of the input (like reordering the input equations, or disabling the probabilistic representation of affine cells) to improve the timing.
      \item[\faClose] Bertini terminated the computation with an error.
      \item[\faWarning] The result given by Bertini is not consistent with our result in terms of degree/dimension.
    \end{description}
\end{table*}
 
% \scriptsize{
% \begin{table*}[p]
%   \caption{Comparing probabilistic and non-probabilistic version}
%   \label{tbl:hyperplane}
%   \begin{minipage}{1.0\linewidth}
%     \begin{center}
% \begin{tabular}{l?rr}
%  & With prob. intersection test & Without prob. intersection\\
% \hline
% dgp6 & 0.79 & 0.36\\
% sys2353 & 7.32 & 2.29\\
% sys2874 & 9.00 & 2.77\\
% sys2880 & 2.80 & 5.58\\
% \hline
% \end{tabular}
%   \end{center}
%   \end{minipage}
% \end{table*}}
%\balance

% \balance
\clearpage
\printbibliography
\end{document}